\documentclass[12pt]{article}
\usepackage[utf8]{inputenc}
\usepackage[T1]{fontenc}
\usepackage{array}
\usepackage{amsmath, amsthm}
\usepackage{amssymb}
\usepackage{enumitem}

\newtheorem{theorem}{Theorem}[section]
\newtheorem{proposition}[theorem]{Proposition}
\newtheorem{corollary}[theorem]{Corollary}
\newtheorem{conjecture}[theorem]{Conjecture}
\newtheorem{lemma}[theorem]{Lemma}

\newtheorem{definition}[theorem]{Definition}

\newtheorem{remark}[theorem]{Remark}

\newcommand{\Factor}{F}
\newcommand{\pfl}{L}
\newcommand{\resalpha}{5}
\DeclareMathOperator{\Prefix}{Prf}
\DeclareMathOperator{\Suffix}{Suf}
\DeclareMathOperator{\occur}{occur}

\begin{document}

\title{Restivo Salemi property for $\alpha$-power free languages with $\alpha\geq 5$ and $k\geq 3$ letters}

\author{Josef Rukavicka\thanks{Department of Mathematics,
Faculty of Nuclear Sciences and Physical Engineering, Czech Technical University in Prague
(josef.rukavicka@seznam.cz).}}

\date{\small{April 28, 2025}}

\maketitle

\begin{abstract}
In 2009, Shur published the following conjecture: \emph{Let $L$ be a power-free language and let $e(L)\subseteq L$ be the set of words of $L$ that can be extended to a bi-infinite word respecting the given power-freeness. If $u, v \in e(L)$ then $uwv \in e(L)$ for some word $w$.}

Let $L_{k,\alpha}$ denote an $\alpha$-power free language over an alphabet with $k$ letters, where $\alpha$ is a positive rational number and $k$ is positive integer. We prove the conjecture for the languages $L_{k,\alpha}$, where $\alpha\geq 5$ and $k\geq 3$.
\end{abstract}

\section{Introduction}
Due to a significant overlap of notation, we copy and use the definitions and some explications from sections  ``Introduction'' and ``Preliminaries'' from \cite{10.1007/978-3-031-34326-1_12}.

Let $r$ be a non-empty word, let $\alpha$ be a positive rational number, and let $r^{\alpha}=rr\cdots rt$ be such that $\frac{\vert r^{\alpha}\vert}{\vert r\vert}=\alpha$ and $t$ is a prefix of $r$. We call the word $r^{\alpha}$ an $\alpha$-\emph{power} of $r$. For example $(1234)^{3}=123412341234$, $(1234)^{\frac{7}{4}}=1234123$, and $(1234)^{\frac{1}{2}}=12$.

Suppose a finite or infinite word $w$. Let \[\begin{split}\Theta(w)=\{(r,\alpha)\mid r^{\alpha}\mbox{ is a factor of }w \mbox{ and }r\mbox{ is a non-empty word and }\\ \alpha\mbox{ is a positive rational number}\}\mbox{.}\end{split}\]
We say that $w$ is $\alpha$-\emph{power-free} if \[\{(r,\beta)\in\Theta(w)\mid \beta\geq \alpha\}=\emptyset\] 
and we say that $w$ is $\alpha^+$-power-free if \[\{(r,\beta)\in\Theta(w)\mid \beta>\alpha\}=\emptyset\mbox{.}\] 

The square free ($2$-power-free) and cube free words ($3$-power-free) are well known examples of the power-free words. More detailed introduction into the field of power-free words can be found, for instance, in \cite{Rampersad_Narad2007} and \cite{10.1007/978-3-642-22321-1_3}.

In 1985, Restivo and Salemi presented a list of five problems concerning the extendability of power-free words \cite{10.1007/978-3-642-82456-2_20}. Problem $4$ and Problem $5$ are relevant for the current article, hence we recall these two problems:
\begin{itemize} 
\item
Problem $4$:  Given finite $\alpha$-power-free words $u$ and $v$, decide whether there is a transition word $w$, such that $uwu$ is $\alpha$-power-free.
\item
Problem $5$: Given finite $\alpha$-power-free words $u$ and $v$, find a transition word $w$, if it exists.
\end{itemize}
In 2019, a survey on the solution of all the five problems has been presented in \cite{10.1007/978-3-030-19955-5_27}; in particular, Problem $4$ and Problem $5$ were resolved for some binary languages. In addition, in \cite{10.1007/978-3-030-19955-5_27} the authors resolved Problem $4$ and Problem $5$ for cube free words.

Let $\mathbb{N}$ denote the set of positive integers and let $\mathbb{Q}$ denote the set of rational numbers.
\begin{definition}(\cite[Definition $1$]{10.1007/978-3-030-48516-0_22})
Let 
\[
\begin{split}
\Upsilon=
\{(k,\alpha)\mid k\in \mathbb{N}\mbox{ and }\alpha\in \mathbb{Q}\mbox{ and }k=3 \mbox{ and }\alpha>2\}\\ \cup\{(k,\alpha)\mid k\in \mathbb{N}\mbox{ and }\alpha\in \mathbb{Q}\mbox{ and }k>3\mbox{ and }\alpha\geq 2\}\\ \cup\{(k,\alpha^+)\mid k\in \mathbb{N}\mbox{ and }\alpha\in \mathbb{Q}\mbox{ and } k\geq3\mbox{ and }\alpha\geq 2\}\mbox{.}
\end{split}
\]
\end{definition}
\begin{remark}(\cite[Remark $1$]{10.1007/978-3-030-48516-0_22})
The definition of $\Upsilon$ says that:
If $(k,\alpha)\in \Upsilon$ and $\alpha$ is a number with $+$ then $k\geq 3$ and $\alpha\geq 2$.
If $(k,\alpha)\in \Upsilon$ and $\alpha$ is just a number then $k=3$ and $\alpha>2$ or $k>3$ and $\alpha\geq 2$.
\end{remark}

In 2020, Problem $4$ and Problem $5$ were resolved for $\alpha$-power-free languages over an alphabet with $k$ letters, where $(k,\alpha)\in\Upsilon$ \cite{10.1007/978-3-030-48516-0_22}.

In 2009, a conjecture related to Problems $4$ and Problem $5$ of Restivo and Salemi appeared in \cite{10.1007/978-3-642-02737-6_38}: 
\begin{conjecture} \label{rhju558rf8ui}Let $L$ be a power-free language and let $e(L)\subseteq L$ be the set of words of $L$ that can be extended to a bi-infinite word respecting the given power-freeness. If $u, v \in e(L)$ then $uwv \in e(L)$ for some word $w$. \end{conjecture} 
In 2018, Conjecture \ref{rhju558rf8ui} appeared also in \cite{SHALLIT201996} using a ``Restivo Salemi property''; it was defined that a language $L$ has the \emph{Restivo Salemi property} if Conjecture \ref{rhju558rf8ui} holds for the language $L$. (Thus Restivo-Salemi property is exactly the transition property for extendable power-free languages.)

Let $\Sigma_k$ denote an alphabet with $k$ letters. 
Let $L_{k,\alpha}\subseteq \Sigma_k^*$ denote the $\alpha$-power-free language, where $\alpha$ is a positive rational number and $k$ is positive integer. We have that $\pfl_{k,\alpha}=\{w\in\Sigma_k^*\mid w \mbox{ is }\alpha\mbox{-power-free}\}$.

In \cite[Remark $41$]{SHALLIT201996}, it was mentioned that for $2^+\leq \alpha\leq\frac{7}{3}$ the language $\pfl_{2,\alpha}$ has the Restivo Salemi property (Conjecture \ref{rhju558rf8ui} holds for $\pfl_{2,\alpha}$). Recall that $2^+$-power-free language is the overlap free language. Otherwise, as to our best knowledge, Conjecture \ref{rhju558rf8ui} remains open.

Let 
\[\widetilde\Upsilon=
\{(k,\alpha)\mid k\in \mathbb{N}\mbox{ and }\alpha\in \mathbb{Q}\mbox{ and }k\geq 3 \mbox{ and }\alpha\geq 5\}\mbox{.}\]
Clearly $\widetilde \Upsilon\subset \Upsilon$. In the current article, we prove that if $(k,\alpha)\in\widetilde\Upsilon$ then Conjecture \ref{rhju558rf8ui} holds for the language $\pfl_{k,\alpha}$; see Corollary \ref{dj99813jf}. The requirement that \(\alpha\geq 5\) was necessary in the proof of Proposition \(3\) in \cite{10.1007/978-3-031-34326-1_12}. Since we use the results of \cite{10.1007/978-3-031-34326-1_12}, we require also for the current paper that \(\alpha\geq 5\).

In order to prove our result, we apply some ideas and results from \cite{10.1007/978-3-030-48516-0_22} and \cite{10.1007/978-3-031-34326-1_12}. 
We briefly explain the very basic idea. If $u,v$ are square power-free words ($2$-power-free) and $x$ is a letter such that $x$ is a factor of neither $u$ nor $v$, then clearly $uxv$ is square free. Just note that there cannot be a factor in $uxv$ which is a square and contains $x$, because $x$ has only one occurrence in $uxv$. 

This very basic idea was generalized to a finite number of occurrences of \(x\) instead of a one single occurrence of \(x\): In \cite{10.1007/978-3-030-48516-0_22}, the author constructed right and left infinite words containing a given factor and a non-recurrent letter $x$. The non-recurrence of the letter $x$ and several involved observations of recurrent and non-recurrent factors allowed the author to resolve Problem \(4\) and  Problem \(5\) of Restivo and Salemi for languages \(\pfl_{k,\alpha}\), where \((k,\alpha)\in\Upsilon\). 

In 2023, the construction of right and left infinite power-free words with a non-recurrent letter has been extended to bi-infinite $\alpha$-power-free words over an alphabet with $k$-letters, where $(k,\alpha)\in\widetilde \Upsilon$ \cite{10.1007/978-3-031-34326-1_12}. In \cite{10.1007/978-3-031-34326-1_12}, it was also predicted that these bi-infinite power-free words with a non-recurrent letter could help to prove Conjecture \ref{rhju558rf8ui}. In the current article we apply indeed these bi-infinite $\alpha$-power-free words with a non-recurrent letter \(x\) to prove Conjecture \ref{rhju558rf8ui} for languages \(\pfl_{k,\alpha}\), where \((k,\alpha)\in\widetilde \Upsilon\).

\section{Preliminaries}

Let $\Sigma_k^+$ denote the set of all non-empty finite words over $\Sigma_k$, let $\epsilon$ denote the empty word, let $\Sigma_k^*=\Sigma_k^+\cup\{\epsilon\}$, let $\Sigma_k^{\mathbb{N},R}$ denote the set of all right infinite words over $\Sigma_k$, let $\Sigma_k^{\mathbb{N},L}$ denote the set of all left infinite words over $\Sigma_k$, and let $\Sigma_k^{\mathbb{Z}}$ denote the set of all bi-infinite words over $\Sigma_k$.

Let $\Sigma_k^{\infty}=\Sigma_k^{\mathbb{N},L}\cup \Sigma_k^{\mathbb{N},R}\cup \Sigma_k^{\mathbb{Z}}$. We call $w\in \Sigma_k^{\infty}$ an infinite word.

Let $\occur(w,t)$ denote the number of all occurrences of the non-empty factor $t\in \Sigma_k^+$ in the word $w\in \Sigma_k^*\cup\Sigma_k^{\infty}$. If $w\in \Sigma_k^{\infty}$ and $\occur(w,t)=\infty$, then we call $t$ a \emph{recurrent} factor in $w$.

Let $\Factor(w)$ denote the set of all finite factors of a finite or infinite word $w\in \Sigma_k^*\cup\Sigma_k^{\infty}$. The set $\Factor(w)$ contains the empty word and if $w$ is finite then also $w\in \Factor(w)$. Let $\Factor_r(w)\subseteq \Factor(w)$ denote the set of all recurrent non-empty factors of $w\in \Sigma_k^{\infty}$.

Let $\Prefix(w)\subseteq \Factor(w)$ denote the set of all prefixes of $w\in \Sigma_k^*\cup\Sigma_k^{\mathbb{N},R}$ and let $\Suffix(w)\subseteq \Factor(w)$ denote the set of all suffixes of $w\in \Sigma_k^*\cup\Sigma_k^{\mathbb{N},L}$. We define that $\epsilon\in \Prefix(w)\cap\Suffix(w)$ and if $w$ is finite then also $w\in \Prefix(w)\cap\Suffix(w)$. 

Let $\pfl_{k,\alpha}^{\infty}\subseteq\Sigma_k^{\infty}$ denote the set of all infinite $\alpha$-power-free words over $\Sigma_k$; formally $\pfl_{k,\alpha}^{\infty}=\{w\in \Sigma_k^{\infty}\mid\Factor(w)\subseteq \pfl_{k,\alpha}\}$. In addition we define $\pfl_{k,\alpha}^{\mathbb{N},R}=\pfl_{k,\alpha}^{\infty}\cap\Sigma_k^{\mathbb{N},R}$, $\pfl_{k,\alpha}^{\mathbb{N},L}=\pfl_{k,\alpha}^{\infty}\cap\Sigma_k^{\mathbb{N},L}$, and $\pfl_{k,\alpha}^{\mathbb{Z}}=\pfl_{k,\alpha}^{\infty}\cap\Sigma_k^{\mathbb{Z}}$; it means the sets of right infinite, left infinite, and bi-infinite $\alpha$-power-free words, respectively.

We define the \emph{reversal} $w^R$ of a finite or infinite word $w=\Sigma_k^*\cup\Sigma_k^{\mathbb{N},R}\cup\Sigma_k^{\mathbb{N},L}$ as follows: \begin{itemize}
\item $\epsilon^R=\epsilon$.
\item If $w\in\Sigma_k^+$ and $w=w_1w_2\dots w_m$, where $w_i\in \Sigma_k$ and $1\leq i\leq m$, then $w^R=w_mw_{m-1}\dots w_2w_1$. 
\item If $w\in \Sigma_k^{\mathbb{N},L}$ and $w=\dots w_2w_1$, where $w_i\in \Sigma_k$ and $i\in \mathbb{N}$, then $w^R=w_1w_2\dots\in \Sigma_k^{\mathbb{N},R}$. \item If $w\in \Sigma_k^{\mathbb{N},R}$ and $w=w_1w_2\dots$, where $w_i\in \Sigma_k$ and $i\in \mathbb{N}$, then $w^R=\dots w_2w_1\in \Sigma_k^{\mathbb{N},L}$. 
\end{itemize}

\begin{remark}
It is obvious that the reverse function preserves the power-freeness and that every factor of an $\alpha$-power-free word is also $\alpha$-power-free.
\end{remark}

The next proposition is a reformulation of Corollary $1$ from \cite{10.1007/978-3-030-48516-0_22} using only the notation of the current article.
\begin{proposition}(reformulation of \cite[Corollary $1$]{10.1007/978-3-030-48516-0_22})
\label{nb3rdyd887j}
If $(k,\alpha)\in\Upsilon$, $v\in\pfl_{k,\alpha}^{\mathbb{N},L}$, $z\in\Suffix(v)$, $x\in\Factor_r(v)\cap\Sigma_k$, $s\in\pfl_{k,\alpha}^{\mathbb{N},L}$, and $x\not\in\Factor(s)$ then there is a finite word $u\in\Sigma_k^*$ such that  $z\in\Suffix(su)$ and $su\in\pfl_{k,\alpha}^{\mathbb{N},L}$.
\end{proposition}

\begin{remark}(\cite[Remark \(3\)]{10.1007/978-3-031-34326-1_12})
Proposition \ref{nb3rdyd887j} says that if $z$ is a finite \(\alpha\)-power-free word that can be extended to a left infinite \(\alpha\)-power-free word having a letter $x$ as a recurrent factor and $s$ is a left infinite \(\alpha\)-power-free word not containing the letter $x$ as a factor, then there is a left infinite \(\alpha\)-power-free word containing $z$ as a suffix and having only a finite number of occurrences of $x$.
\end{remark}

The next elementary lemma follows easily from famous theorems of Thue \cite{thue1906ueber,thue1912gegenseitige}.

\begin{lemma}(reformulation of \cite[Lemma $2$]{10.1007/978-3-030-48516-0_22})
\label{dy77ejhfiffu}
If $k\geq 3$ and $\alpha>2$ then $\pfl_{k-1,\alpha}^{\mathbb{N},R}\not=\emptyset$.
\end{lemma}

Given $t\in\Sigma_k^{\mathbb{N},R}$, let
\(\Phi(t)=\{s\in \Sigma_{k}^{\mathbb{N},L}\mid \Factor(s)\subseteq \Factor_r(t)\cup\{\epsilon\}\}\mbox{.}\) 
It means that all non-empty factors of $s\in \Phi(t)$ are recurrent factors of $t$. The following lemma is an immediate consequence of König's lemma \cite{Koenig1926}.
\begin{lemma}(\cite[Lemma $4$]{10.1007/978-3-030-48516-0_22})
\label{lrktikl009iu8}
If \(k\in\mathbb{N}\) and $t\in \Sigma_k^{\mathbb{N},R}$ then $\Phi(t)\not=\emptyset$.
\end{lemma}

The main result of \cite{10.1007/978-3-031-34326-1_12} can be formulated as follows:
\begin{theorem}
\label{yud7uej41}
If \((k,\alpha)\in\widetilde\Upsilon\), $v\in L_{k,\alpha}^{\mathbb{Z}}$, $w\in\Factor(v)\setminus\{\epsilon\}$, then there are $\widetilde v\in L_{k,\alpha}^{\mathbb{Z}}$ and a letter $x\in\Sigma_k$ such that $w\in\Factor(\widetilde v)$, $x\in\Factor(w)$, and $\occur(w,x)<\infty$.
\end{theorem}
\begin{remark}
We will not apply Theorem \ref{yud7uej41} directly, because we will need some additional properties of the constructed bi-infinite word. Instead, we will apply \cite[Theorem $1$]{10.1007/978-3-031-34326-1_12} (see Theorem \ref{d78fju5e4} in the current article).
\end{remark}

\section{Bi-infinite $\alpha$-power words}

For the rest of the article suppose that $(k,\alpha)\in\widetilde\Upsilon$; it means $k\geq 3$ and $\alpha\geq \resalpha$.
We define two technical sets $\Gamma$ and $\Delta$.
\begin{definition}(\cite[Definition $2$]{10.1007/978-3-031-34326-1_12})
Let $\Gamma$ be a set of triples such that $(w,\eta ,u)\in\Gamma$ if and only if \begin{itemize}
\item $w\in\Sigma_k^+$, $\eta ,u\in\Sigma_k^*$,  and 
\item if $\vert u\vert\leq \vert w\vert$ then $\vert \eta \vert\geq(\alpha+1)\alpha^{\vert w\vert-\vert u\vert}\vert w\vert\mbox{.}$
\end{itemize}
\end{definition}
\begin{remark}(\cite[Remark $4$]{10.1007/978-3-031-34326-1_12})
The set $\Gamma$ contains triples of finite words $w,\eta,u$ such that $w$ is non-empty, and if $u$ is shorter than $w$, then the word $\eta$ is ``sufficiently'' longer than $w$. 
\end{remark}
\begin{definition}(\cite[Definition $3$]{10.1007/978-3-031-34326-1_12})
\label{ccn512rf1x}
Let $\Delta$ be a set of $6$-tuples such that \[(s,\sigma ,w,\eta ,x,u)\in\Delta\] if and only if
\begin{enumerate}
\item $s\in\Sigma_{k}^{\mathbb{N}, L}$, $\sigma ,\eta ,u\in\Sigma_k^*$, $w\in\Sigma_k^+$, $x\in\Sigma_k$, 
\item \label{du77bxn21b} $s\sigma w\eta xu\in\pfl_{k,\alpha}^{\mathbb{N},L}$,
\item \label{du87ejh14} $(w,\eta ,u)\in\Gamma$,
\item \label{ddhy7vzlp} $\occur(s\sigma w,w)=1$, and 
\item $x\not\in\Factor(s)\cup\Factor(u)$.
\end{enumerate}
\end{definition}

In \cite{10.1007/978-3-031-34326-1_12}, it was shown that if $(s,\sigma ,w,\eta ,x,\epsilon)\in\Delta$ and $t$ is a right infinite $\alpha$-power-free word with no occurrence of the letter $x$ then there is a bi-infinite $\alpha$-power-free word containing the factor $w$ and having only a finite number of occurrences of $x$:
\begin{theorem}(\cite[Theorem $1$]{10.1007/978-3-031-34326-1_12})
\label{d78fju5e4}
If $(s,\sigma ,w,\eta ,x,\epsilon)\in\Delta$, $t\in\pfl_{k,\alpha}^{\mathbb{N},R}$, and $x\not\in\Factor(t)$ then there is $\widehat \eta \in\Prefix(\eta )$ such that $s\sigma w\widehat \eta xt\in\pfl_{k,\alpha}^{\mathbb{Z}}$.
\end{theorem}
\begin{definition}
\label{dxhb7d6djwb}(\cite[Definition $4$]{10.1007/978-3-031-34326-1_12})
Suppose $v\in\pfl_{k,\alpha}^{\mathbb{Z}}$ and $w\in\Factor(v)\setminus\{\epsilon\}$. If there are $v_1\in\pfl_{k,\alpha}^{\mathbb{N},L}$ and $v_2\in\pfl_{k,\alpha}^{\mathbb{N},R}$ such that $v=v_1v_2$ and $w\in\Factor_r(v_2)$ then we say that $w$ is \emph{on-right-side recurrent} in $v$.
\end{definition}
\begin{remark}\cite[Remark $8$]{10.1007/978-3-031-34326-1_12})
Note in Definition \ref{dxhb7d6djwb} that no restriction is imposed on the recurrence of $w$ in $v_1$. This means that $w$ may also be recurrent in $v_1$.
\end{remark}

We show how to construct an element of $\Delta$ from a bi-infinite $\alpha$-power-free word $v$ with a finite factor $w$ containing the letter \(x\).
\begin{lemma}
\label{udi8ds0fne5df}
If $v\in\pfl_{k,\alpha}^{\mathbb{Z}}$, $w\in\Factor(v)\setminus\{\epsilon\}$, and $x\in\Factor(w)\cap\Sigma_k$ then there are $\overline s\in\pfl_{k,\alpha}^{\mathbb{N},L}$, and $\sigma ,\eta \in\Sigma_k^*$ such that $(\overline s,\sigma, w,\eta, x,\epsilon)\in\Delta$.
\end{lemma}
\begin{proof}
We distinguish following cases:
\begin{enumerate}
\item \label{dhu78djcn} If \(\occur(v,w)=\infty\) or \(x\) is on-right-side recurrent in \(v\) then it is easy to see that 
 there are $\overline v\in\pfl_{k,\alpha}^{\mathbb{N},L}$, $\widehat v\in\pfl_{k,\alpha}^{\mathbb{N},R}$,  and $\eta \in\Sigma_k^*$ such that $v=\overline vw\eta x\widehat v$ and $(w,\eta ,\epsilon)\in\Gamma$.
\item Otherwise; it means that  \(\occur(v,w)<\infty\) and \(x\) is not on-right-side recurrent in \(v\).
Then there are $\overline v\in\pfl_{k,\alpha}^{\mathbb{N},L}$, $\widehat v\in\pfl_{k,\alpha}^{\mathbb{N},R}$ and $\overline \eta \in\Sigma_k^*$ such that $v=\overline  vw\overline \eta \widehat v$,   
\begin{equation}\label{hghd7898e}x\not\in\Factor(\widehat v)\mbox{, }\end{equation}
and 
\begin{equation}\label{ndhd8773}\occur(\overline vw,w)=1\mbox{.}\end{equation}

Let $\eta_0\in\Prefix(\widehat v)$ be such that \begin{equation}\label{u9ddbh3ff3}\vert \eta_0\vert>\max\{\vert w\overline \eta \vert,(\alpha+1)\alpha^{\vert w\vert}\vert w\vert\}\mbox{.}\end{equation} We claim that $\overline vw\overline \eta \eta_0x\in\pfl_{k,\alpha}^{\mathbb{N},L}$. To get a contradiction, suppose that there are $r$ and $\beta>\alpha$ such that $r^{\beta}\in\Suffix(\overline vw\overline \eta \eta_0x)$. From (\ref{hghd7898e}) we have that $x\not\in\Factor(\eta_0)$; consequently $\eta_0x\in\pfl_{k,\alpha}$ and $\eta_0x\in\Suffix(r)$. From (\ref{u9ddbh3ff3}) we have that $\vert \eta_0\vert>
\vert w\overline \eta \vert$; it follows that $w\in\Factor(rr)$. From $\beta>\alpha \geq 5$ and (\ref{ndhd8773}), we conclude that are not such $r$ and $\beta$. Set $\eta =\overline \eta \eta_0$. From (\ref{u9ddbh3ff3}) it is obvious $(w,\eta ,\epsilon)\in\Gamma$. 
\end{enumerate}
In both cases we found $\overline v$ and $\eta $ such that $\overline vw\eta x\in\pfl_{k,\alpha}^{\mathbb{N},L}$, and $(w,\eta ,\epsilon)\in\Gamma$.

\begin{itemize}
\item If $x\not\in\Factor(\overline v)$ then let $\overline s=\overline v$. It is obvious  that $(\overline s,\epsilon, w,\eta, x,\epsilon)\in\Delta$.
\item If $x\in\Factor(\overline v)$ then let $t$ be a right infinite $\alpha$-power-free word on the alphabet $\Sigma_k\setminus\{x\}$. Since $(k,\alpha)\in\widetilde\Upsilon$, Lemma \ref{dy77ejhfiffu} asserts that $t$ exists. Let $s=t^R$. Clearly $s$ is a left infinite $\alpha$-power-free word and $x\not\in\Factor(s)$.

Let $z=w\eta x$. Proposition \ref{nb3rdyd887j} implies that there is a finite word $u$ such that $su\in\pfl_{k,\alpha}^{\mathbb{N},L}$ and $z\in\Suffix(su)$. Obviously there are $\overline s\in\pfl_{k,\alpha}^{\mathbb{N},L}$ and $\sigma,\eta_0 \in\Sigma_k^*$ such that \(\overline s\sigma w\eta_0 x=su\), \(\eta\in\Suffix(\eta_0)\), \(w\eta x\in\Suffix(\overline s\sigma w\eta_0 x)\), and $(\overline s,\sigma, w,\eta_0, x,\epsilon)\in\Delta$. Note that  since $x\in\Factor(w)$ and $x\not\in\Factor(\overline s)$  Property \ref{ddhy7vzlp} of Definition \ref{ccn512rf1x} is easy to be asserted by a proper choice of $\sigma$ and \(\eta_0\).
\end{itemize}

This completes the proof. 

\end{proof}

Suppose two bi-infinite $\alpha$-power-free words $s_1\widetilde w_1t_1$ and $s_2\widetilde w_2t_2$, where $s_1,s_2$ are left infinite words, $t_1,t_2$ are right infinite words, $\widetilde w_1, \widetilde w_2$ are finite words, the factors of $s_2$ are recurrent factors of $t_1$, and there is a letter $x$ such that $x$ is not a factor of $s_1,s_2,t_1,t_2$ and $x$ is a factor of $\widetilde w_2$. We show that there is a finite word $w$ such that $s_1\widetilde w_1w\widetilde w_2t_2$ is a bi-infinite $\alpha$-power-free word.
\begin{proposition}
\label{mr7e8kjif}
If $\widetilde w_1,\widetilde w_2\in \Sigma_k^+$, $s_1,s_2\in\Sigma_k^{\mathbb{N},L}$, $t_1,t_2\in\Sigma_k^{\mathbb{N},R}$, $x\in\Sigma_k$, $s_1\widetilde w_1t_1, s_2\widetilde w_2t_2\in\pfl_{k,\alpha}^{\mathbb{Z}}$, $\Factor(s_2)\subseteq \Factor_r(t_1)$, $x\not\in\Factor(s_1)\cup\Factor(s_2)\cup\Factor(t_1)\cup\Factor(t_2)$, and $x\in\Prefix(\widetilde w_2)$ then there is $w\in\Prefix(t_1)$ such that $s_1\widetilde w_1w\widetilde w_2t_2\in\pfl_{k,\alpha}^{\mathbb{Z}}$.
\end{proposition}
\begin{proof}
Let $p\in \Suffix(s_2)$ be the shortest suffix of \(s_2\) such that \begin{equation}\label{dju8w0dj9sj9}\vert p\vert>\max\{\vert \widetilde w_1\vert,\vert \widetilde w_2\vert\}\mbox{.}\end{equation} Let $h\in \Prefix(t_1)$ be the shortest prefix of $t_1$ such that $hp\in \Prefix(t_1)$ and $\vert h\vert>\vert p\vert$; such $h$ exists, because $\Factor(s_2)\subseteq \Factor_r(t_1)$.
We have that $s_1\widetilde w_1hp\in \pfl_{k,\alpha}^{\mathbb{N},L}$, since $hp\in \Prefix(t_1)$ and $s_1\widetilde w_1t_1\in \pfl_{k,\alpha}^{\mathbb{Z}}$. We show that $s_1\widetilde w_1hp\widetilde w_2t_2\in \pfl_{k,\alpha}^{\mathbb{Z}}$.

To get a contradiction, suppose that there are $g\in \Prefix(\widetilde w_2t_2)$, $y\in\Sigma_k$, $r\in \Sigma_k^+$, and $\beta> \alpha$ such that $r^{\beta}\in \Suffix(s_1\widetilde w_1hpgy)$, $gy\in\Prefix(\widetilde w_2t_2)$, and $s_1\widetilde w_1hpg\in\pfl_{k,\alpha}^{\mathbb{N},L}$. 
Since $p\widetilde w_2t_2\in\pfl_{k,\alpha}^{\mathbb{N},R}$ and $x\in\Prefix(\widetilde w_2)$ we have that $\vert r^{\beta}\vert>\vert pgy\vert$ and $x\in\Factor(r)$.

Let $\widehat w_2\in\Sigma_k^*$ be such that $\widehat w_2x\in\Prefix(\widetilde w_2)$ and $x\not\in\Factor(\widehat w_2)$. Clearly $\widehat w_2$ exists and is uniquely determined; realize that $\widehat w_2$ may be the empty word. 

Since $p\widehat w_2x\in\Factor(r^{\beta})$ and $x\not\in\Factor(p\widehat w_2)$, 
we have that \begin{equation}\label{djfuf8ejj8}p\widehat w_2x\in\Factor(rr)\mbox{.}\end{equation}
Recall $x\not\in\Factor(s_1)\cup\Factor(hp)\cup \Factor(t_2)$. Then it follows from (\ref{dju8w0dj9sj9}) that \begin{equation}\label{dy66jiif89ee}\occur(s_1\widetilde w_1hp\widetilde w_2t_2,p\widehat w_2x)\leq 2\mbox{.}\end{equation}

From (\ref{djfuf8ejj8}), (\ref{dy66jiif89ee}), and $\beta>\alpha\geq 5$ we conclude that 
there are no such $r,g,\beta$ and hence $s_1\widetilde w_1hp\widetilde w_2t_2\in \pfl_{k,\alpha}^{\mathbb{Z}}$. 

This completes the proof.
\end{proof}

The main theorem of the article says that if $w_3,w_4$ are finite non-empty factors of bi-infinite $\alpha$-power-free words then there is a bi-infinite $\alpha$-power-free word $v$ and a finite word $w_0$ such that $w_3w_0w_4$ is a factor of $v$.
\begin{theorem}
\label{nnbdh8de}
If $v_1, v_2\in\pfl_{k,\alpha}^{\mathbb{Z}}$, $w_3\in\Factor(v_1)\setminus\{\epsilon\}$, and $w_4\in\Factor(v_2)\setminus\{\epsilon\}$ then there are $w_0\in\Sigma_k^*$ and $v\in\pfl_{k,\alpha}^{\mathbb{Z}}$ such that $w_3w_0w_4\in\Factor(v)$.
\end{theorem}
\begin{proof}
If $\Factor(v_1)\cap\Factor(v_2)=\emptyset$ then let $v_{11}$, $v_{12}$, $v_{21}$, and $v_{22}$ be such that $v_1=v_{11}v_{12}$, $v_2=v_{21}v_{22}$, $w_3\in\Factor(v_{11})$, and $w_4\in\Factor(v_{22})$. It is clear $v_{11}{v_{22}}$ is $\alpha$-power-free since $\Factor(v_{11})\cap\Factor(v_{22})=\emptyset$. The theorem follows.

Thus suppose that $\Factor(v_1)\cap\Factor(v_2)\not=\emptyset$ and let $x\in\Factor(v_1)\cap\Factor(v_2)\cap\Sigma_k$.
Let \(w_1\in\Factor(v_1)\) and \(w_2\in\Factor(v_2)\) be such that \(x\in\Factor(w_1)\cap\Factor(w_2)\), \(w_3\in\Factor(w_1)\), and \(w_4\in\Factor(w_2)\). It is clear that such \(w_1,w_2\) exist.

Lemma \ref{udi8ds0fne5df} asserts that there are $s_1\in\pfl_{k,\alpha}^{\mathbb{N},L}$ and $\sigma_1 ,\eta_1\in\Sigma_k^*$ such that $(s_1,\sigma_1, w_1,\eta_1, x,\epsilon)\in\Delta$.
Lemma \ref{udi8ds0fne5df} asserts also that there are $t_2\in\pfl_{k,\alpha}^{\mathbb{N},R}$ and $\sigma_2,\eta_2\in\Sigma_{k}^*$ such that $(t_2^R,\sigma_2^R, w_2^R,\eta_2^R, x,\epsilon)\in\Delta$. 

Let $t_1\in\pfl_{k-1,\alpha}^{\mathbb{N},R}$ be a right infinite $\alpha$-power-free word on the alphabet $\Sigma_k\setminus\{x\}$. Since $\alpha\geq 5$, Lemma \ref{dy77ejhfiffu} asserts that $t_1$ exists. Let $s_2\in\Phi(t_1)$; Lemma \ref{lrktikl009iu8} implies that $\Phi(t_1)\not=\emptyset$ and consequently \(s_2\) exists. Clearly $s_2$ is $\alpha$-power-free and $x\not\in\Factor(s_2)$.

Theorem \ref{d78fju5e4} implies that there are $\widehat \eta_1\in\Prefix(\eta_1)$ and $\widehat \eta_2\in\Suffix(\eta_2)$ such that $s_1\sigma_1w_1\widehat \eta_1xt_1, t_2^R\sigma_2^Rw_2^R\widehat \eta_2^R xs_2^R\in\pfl_{k,\alpha}^{\mathbb{Z}}$. Since the reverse operation preserves the power-freeness we have also that 
$s_2x\widehat \eta_2w_2\sigma_2t_2\in\pfl_{k,\alpha}^{\mathbb{Z}}$.

Let $\widetilde w_1=\sigma_1w_1\widehat \eta_1x$ and let $\widetilde w_2=x\widehat \eta_2w_2\sigma_2$. Then Proposition \ref{mr7e8kjif} implies that there is $w\in\Sigma_k^*$ such that $s_1\widetilde w_1 w\widetilde w_2t_2\in\pfl_{k,\alpha}^{\mathbb{Z}}$. 

Table \ref{hsud78v225} shows the structure of the resulting word using also the notation from the proof of Proposition \ref{mr7e8kjif}.

\begin{table}
\centering
\begin{tabular}{|l|cccc|l|l|llll|l|}
\hline
                            & \multicolumn{4}{c|}{$\widetilde w_1$}                                                                         &     &     & \multicolumn{4}{c|}{$\widetilde w_2$}                                                                         &       \\ \hline
\multicolumn{1}{|c|}{$s_1$} & \multicolumn{1}{c|}{$\sigma_1$} & \multicolumn{1}{c|}{$w_1$} & \multicolumn{1}{c|}{$\widehat \eta_1$} & $x$ & $h$ & $p$ & \multicolumn{1}{l|}{$x$} & \multicolumn{1}{l|}{$\widehat \eta_2$} & \multicolumn{1}{l|}{$w_2$} & $\sigma_2$ & $t_2$ \\ \hline
\end{tabular}
\caption{The structure of the word $v$.}
\label{hsud78v225}
\end{table}

Since \(w_3\in\Factor(w_1)\), \(w_4\in\Factor(w_2)\), \(w_1\in\Factor(\widetilde w_1)\), and \(w_2\in\Factor(\widetilde w_2)\), the theorem follows.
This completes the proof.
\end{proof}

Theorem \ref{nnbdh8de} has the following obvious corollary.
\begin{corollary}
\label{dj99813jf}
    Conjecture \ref{rhju558rf8ui} holds for the languages $L_{k,\alpha}$, where $\alpha\geq 5$ and $k\geq 3$.
\end{corollary}

\bibliographystyle{siam}
\IfFileExists{biblio.bib}{\bibliography{biblio}}{\bibliography{../!bibliography/biblio}}

\end{document}